\documentclass[a4paper,fleqn]{cas-dc}

\usepackage[numbers]{natbib}
\usepackage{xcolor}
\usepackage{hyperref}
\usepackage{graphicx}
\usepackage{subcaption}

\usepackage{enumerate}
\usepackage{flushend}

\newcommand{\ym}{\mathsf{y}_{\mathsf{m}}}

\newcommand{\yf}{\mathsf{y}_{\mathsf{f}}}

\newcommand{\baryf}{\bar{\mathsf{y}}_{\mathsf{f}}}

\newtheorem{theorem}{Theorem}
\newtheorem{lemma}{Lemma}
\newtheorem{proposition}{Proposition}
\newdefinition{remark}{Remark}
\newdefinition{definition}{Definition}
\newproof{proof}{Proof}
\newproof{pot}{Proof of Theorem \ref{thm}}
\newdefinition{assumption}{Assumption}

\definecolor{myred}{rgb}{0.7,0.1,0.16}
\definecolor{myblue}{rgb}{0,0.32,0.7}
\definecolor{mygreen}{rgb}{0.133,0.545,0.133}

\graphicspath{{figures/}}

\begin{document}
\let\WriteBookmarks\relax
\def\floatpagepagefraction{1}
\def\textpagefraction{.001}

\title [mode = title]{Nonholonomic Source Seeking by Torque Tuning: Local and Semi-Global Feedbacks}  
\shorttitle{Nonholonomic Source Seeking by Torque Tuning: Local and Semi-Global Feedbacks}
\shortauthors{B. Wang}


\author[1]{Bo Wang}\cormark[1]
\ead{bwang1@ccny.cuny.edu}
\credit{}
\affiliation[1]{organization={Department of Mechanical Engineering, The City College of New York, The City University of New York},
            city={New York},
            state={NY 10031},
            country={USA}}


\cortext[1]{Corresponding author.}

\nonumnote{\textit{E-mail addresses:} {bwang1@ccny.cuny.edu} (B. Wang).}

\begin{abstract}
This paper studies source seeking for a torque-controlled nonholonomic vehicle with a laterally displaced scalar sensor. The vehicle has constant forward speed, while its yaw motion is controlled by torque input with unknown inertia and damping. The objective is to steer the vehicle to a source-centered circular motion so that the lateral sensor approaches the unknown source, without using position, heading, source-location, gradient, or source-value information. The proposed torque law combines a fast oscillatory component, which generates averaged steering through symmetric-product approximation, with a slowly tuned bias component, which selects the desired orbit. Two bias-tuning designs are developed. The first is an output-feedback design using only the scalar measurement; it applies a Lie-bracket extremum-seeking update and yields local practical stability. The second is a velocity-assisted design using forward-speed and yaw-rate measurements; it tunes the bias through the yaw-rate tracking error and yields a globally asymptotically stable averaged system, implying semi-global practical stability of the original system. Simulations illustrate the proposed designs.
\end{abstract}


\maketitle

\section{Introduction}
\label{sec:introduction}

Source seeking concerns the problem of steering a dynamical system toward the extremum of an unknown spatial signal using real-time measurements of the signal value \cite{krstic2000stability,ariyur2003real,scheinker2024100}. This problem arises in applications such as locating chemical, thermal, electromagnetic, or acoustic sources, where the source position, signal gradient, and extremal value are not available a priori. Source-seeking techniques have been used in diverse engineering systems, including autonomous vehicles \cite{ZhangArnoldGhodsSiranosianKrstic2007,ZhangSiranosianKrstic2007}, marine vessels \cite{wang2023underactuated}, satellites \cite{wang2025extremum}, and particle accelerators \cite{williams2024experimental}, among others. In this paper, we consider source seeking for a torque-controlled nonholonomic vehicle moving in an unknown static scalar field. 
The vehicle is equipped with a suitable sensor that measures the signal value at the sensor location. The objective is to drive the vehicle toward the source, or equivalently to steer the vehicle toward a source-centered circular motion. Following the standard convention in extremum seeking, the source is treated as the unique minimum of the signal.

Velocity-level nonholonomic source seeking has been studied extensively for kinematic unicycle models, where the forward and angular velocities are treated as directly assignable control inputs. 
One line of work tuned the forward velocity while maintaining a constant angular velocity, using extremum-seeking dithers to recover gradient information without position measurements \cite{ZhangArnoldGhodsSiranosianKrstic2007,frihauf2014single,todorovski2024newton}. A complementary line of work keeps the forward speed constant and tunes the angular velocity to steer the vehicle toward a source-centered orbit \cite{CochranKrstic2009,RaischKrstic2017,DurrKrsticScheinkerEbenbauer2017}. 
Source-seeking schemes that combine forward-speed regulation with steering control have also been developed \cite{GhodsKrstic2010}. 
Related velocity-regulation ideas have also been developed for three-dimensional nonholonomic vehicles by tuning forward, pitch, and yaw velocities \cite{lin20153}. 
These results provide a mature velocity-level theory for nonholonomic source seeking, but they do not address source seeking under the force- or torque-actuated dynamics.

More recently, source seeking has been extended from velocity-level kinematic models to second-order dynamic models, where the inputs are forces or torques rather than velocities. This shift is natural for robotic vehicles, since inertia prevents instantaneous assignment of the forward or angular velocity. A key tool in this direction is the symmetric-product approximation for mechanical systems under large-amplitude, high-frequency inputs \cite{bullo2002averaging}. This framework has enabled extremum- and source-seeking designs for force- and torque-actuated mechanical systems \cite{suttner2022extremum}, including fully-actuated systems on Lie groups \cite{suttner2023extremum} and underactuated surface vessels \cite{wang2023underactuated}.
For nonholonomic vehicles, Suttner and Krsti{\'c} developed force-actuated source-seeking laws for unicycles using symmetric-product approximation \cite{Suttner2019NOLCOS,SuttnerKrstic2020IFAC}. They further extended this approach to torque-actuated nonholonomic source seeking, including planar unicycles and three-dimensional vehicles driven by pitch and yaw torque inputs \cite{suttner2022source,suttner2023nonholonomic}.
Subsequent work addressed local extrema in dynamic-unicycle source seeking via circular sensing and divergence-theorem-based spatial averaging, with delayed measurements used in one design to extract nonlocal descent information \cite{suttner2023nonlocal,suttner2024overcoming}.

The work most closely related to the present paper is \cite{suttner2022source}, which proposed a source-seeking law for a torque-controlled unicycle with a forward-displaced sensor. Their design uses a large-amplitude high-frequency torque input, and the symmetric-product averaging analysis shows that the resulting averaged torque contains gradient-dependent steering information. Under suitable radial-symmetry and growth conditions, the averaged system admits a locally stable circular motion around the source. However, the radius of the limiting orbit is determined by physical and design parameters, and the sensor is not designed to converge to the source location. 
In contrast, the present paper uses a laterally displaced sensor and decomposes the yaw torque into a fast oscillatory component for averaged steering and a slowly tuned bias component for selecting the source-centered orbit. This structure makes the bias torque an explicit tuning target and leads to two designs with different sensing requirements: an output-feedback design using only scalar measurements and a velocity-assisted design using additional speed and yaw-rate measurements.
 
The main contributions of this paper are threefold. First, we introduce a torque-tuning framework for nonholonomic source seeking with a laterally displaced scalar sensor. The yaw torque is separated into a fast oscillatory component for averaged steering and a slowly tuned bias component for selecting the desired source-centered orbit. 
Second, we develop an output-feedback design that uses only the real-time scalar measurement. The bias is tuned through a Lie-bracket extremum-seeking update applied to the steady response of the fast-averaged dynamics, yielding a local practical stability result.
Third, we develop a velocity-assisted design for the case where forward-speed and yaw-rate measurements are available. This design tunes the bias directly through the yaw-rate tracking error, yielding a globally asymptotically stable averaged system and a semi-global practical stability result for the original closed-loop system.

The remainder of the paper is organized as follows. Section \ref{sec:practical_stability} recalls the practical stability notions used in the analysis. Section \ref{sec:problem_statement} formulates the torque-controlled source-seeking problem. Section \ref{sec:control} presents the two proposed torque-tuning designs and establishes the corresponding local and semi-global practical stability results. Section \ref{sec:simulation} illustrates the designs through numerical simulations, and Section \ref{sec:conclusion} concludes the paper.

\section{Practical Stability}
\label{sec:practical_stability}
This section reviews practical stability results for parameter-dependent systems from \cite{moreau2000practical,DurrStankovicEbenbauerJohansson2013,durr2017extremum}, which will be used in the subsequent analysis.
For each $\varepsilon>0$, let
$f^\varepsilon:\mathbb{R}_{\ge 0}\times\mathbb{R}^n\to\mathbb{R}^n$
be a time-dependent vector field, and let
$f:\mathbb{R}^n\to\mathbb{R}^n$
be an autonomous vector field. Consider two systems 
\begin{equation}\label{eq:f-epsilon} 
\dot{x}=f^\varepsilon(t,x),\quad x(t_0)=x_0 
\end{equation} 
and 
\begin{equation}\label{eq:f} 
\dot{\bar{x}}=f(\bar{x}),\quad \bar{x}(t_0)=\bar{x}_0. 
\end{equation}
Assume that, for every $\varepsilon>0$, $t_0\ge 0$, and $x_0\in\mathbb{R}^n$, \eqref{eq:f-epsilon} admits a unique maximal solution defined on $[t_0,\infty)$, and that the same holds for \eqref{eq:f} for every $\bar{x}_0\in\mathbb{R}^n$.
For every $\varepsilon>0$ and $t\ge 0$, let
$\phi_t^\varepsilon:\mathbb{R}^n\to\mathbb{R}^n$
be a bijection. Along the solutions of \eqref{eq:f-epsilon}, consider the time-dependent change of coordinates
\begin{equation}\label{eq:change-of-variables}
\tilde{x}=\phi_t^\varepsilon(x).
\end{equation}
The following definitions and results are adapted from \cite{moreau2000practical}.

\begin{definition}
We say that $x^\star\in\mathbb{R}^n$ is \emph{practically uniformly stable for \eqref{eq:f-epsilon} in the variables \eqref{eq:change-of-variables}} if, for every $\eta>0$, there exist $\delta>0$ and $\varepsilon_0>0$ such that, for every $\varepsilon\in(0, \varepsilon_0]$, every $t_0\ge 0$, and every $\tilde{x}_0\in\mathbb{R}^n$ satisfying $|\tilde{x}_0-x^\star|\le \delta$, the maximal solution $x(t)$ of \eqref{eq:f-epsilon} with $x(t_0)=(\phi_{t_0}^\varepsilon)^{-1}(\tilde{x}_0)$ satisfies $|\phi_t^\varepsilon(x(t))-x^\star|\le \eta$ for all $t\ge t_0$.
\end{definition}

\begin{definition}\label{def:1}
We say that $x^\star\in\mathbb{R}^n$ is \emph{locally practically uniformly asymptotically stable for \eqref{eq:f-epsilon} in the variables \eqref{eq:change-of-variables}} if the following two properties hold.
\begin{enumerate}[1)]
    \item \emph{Practical Uniform Stability.} $x^\star\in\mathbb{R}^n$ is practically uniformly stable for \eqref{eq:f-epsilon} in the variables \eqref{eq:change-of-variables}.
    \item \emph{Local Practical Uniform Attraction.} There exists $r>0$ such that, for every $\eta>0$, there exist $T>0$ and $\varepsilon_0>0$ such that, for every $\varepsilon\in(0, \varepsilon_0]$, every $t_0\ge 0$, and every $\tilde{x}_0\in\mathbb{R}^n$ satisfying $|\tilde{x}_0-x^\star|\le r$, the maximal solution $x(t)$ of \eqref{eq:f-epsilon} with $x(t_0)=(\phi_{t_0}^\varepsilon)^{-1}(\tilde{x}_0)$ satisfies $|\phi_t^\varepsilon(x(t))-x^\star|\le \eta$ for all $t\ge t_0+T$.
\end{enumerate}
\end{definition}

\begin{definition}\label{def:2}
We say that $x^\star\in\mathbb{R}^n$ is \emph{semi-globally practically uniformly asymptotically stable for \eqref{eq:f-epsilon} in the variables \eqref{eq:change-of-variables}} if the following three properties hold.
\begin{enumerate}[1)]
    \item \emph{Practical Uniform Stability.} $x^\star\in\mathbb{R}^n$ is practically uniformly stable for \eqref{eq:f-epsilon} in the variables \eqref{eq:change-of-variables}.

    \item \emph{Semi-Global Practical Uniform Boundedness.} For every $r>0$, there exist $R>0$ and $\varepsilon_0>0$ such that, for every $\varepsilon\in(0, \varepsilon_0]$, every $t_0\ge 0$, and every $\tilde{x}_0\in\mathbb{R}^n$ satisfying $|\tilde{x}_0-x^\star|\le r$, the maximal solution $x(t)$ of \eqref{eq:f-epsilon} with $x(t_0)=(\phi_{t_0}^\varepsilon)^{-1}(\tilde{x}_0)$ satisfies $|\phi_t^\varepsilon(x(t))-x^\star|\le R$ for all $t\ge t_0$.

    \item \emph{Semi-Global Practical Uniform Attraction.} For every $r>0$ and every $\eta>0$, there exist $T>0$ and $\varepsilon_0>0$ such that, for every $\varepsilon\in(0, \varepsilon_0]$, every $t_0\ge 0$, and every $\tilde{x}_0\in\mathbb{R}^n$ satisfying $|\tilde{x}_0-x^\star|\le r$, the maximal solution $x(t)$ of \eqref{eq:f-epsilon} with $x(t_0)=(\phi_{t_0}^\varepsilon)^{-1}(\tilde{x}_0)$ satisfies $|\phi_t^\varepsilon(x(t))-x^\star|\le \eta$ for all $t\ge t_0+T$.
\end{enumerate}
\end{definition}

\begin{definition}
We say that \emph{the solutions of \eqref{eq:f-epsilon} in the variables \eqref{eq:change-of-variables} approximate the solutions of \eqref{eq:f}} if, for every compact set $K\subset\mathbb{R}^n$ and every $\eta,T>0$, there exists $\varepsilon_0>0$ such that, for every $\varepsilon\in(0, \varepsilon_0]$, every $t_0\ge 0$, and every $\bar{x}_0\in K$, the following implication holds: If the solution $\bar{x}(t)$ of \eqref{eq:f} with $\bar{x}(t_0)=\bar{x}_0$ satisfies $\bar{x}(t)\in K$ for all $t\in[t_0,t_0+T]$, then the solution $x(t)$ of \eqref{eq:f-epsilon} with $x(t_0)=(\phi_{t_0}^\varepsilon)^{-1}(\bar{x}_0)$ satisfies $|\phi_t^\varepsilon(x(t))-\bar{x}(t)|\le \eta$ for all $t\in[t_0,t_0+T]$.
\end{definition}

\begin{lemma}\label{lem:1}
    Suppose that the solutions of \eqref{eq:f-epsilon} in the variables \eqref{eq:change-of-variables} approximate the solutions of \eqref{eq:f}. Then the following statements hold.
    \begin{enumerate}
        \item[(i)] If $x^\star$ is a locally asymptotically stable equilibrium of \eqref{eq:f}, then $x^\star$ is locally practically uniformly asymptotically stable for \eqref{eq:f-epsilon} in the variables \eqref{eq:change-of-variables}.

        \item[(ii)] If $x^\star$ is a globally asymptotically stable equilibrium of \eqref{eq:f}, then $x^\star$ is semi-globally practically uniformly asymptotically stable for \eqref{eq:f-epsilon} in the variables \eqref{eq:change-of-variables}.
    \end{enumerate}
\end{lemma}

\section{Problem Statement}
\label{sec:problem_statement}

Consider a nonholonomic mobile robot moving in the plane. Let $(x,y)\in\mathbb R^2$ denote the position of the vehicle center in an inertial frame, let $\theta\in\mathbb R$ denote the heading angle, and let $\omega\in\mathbb R$ denote the yaw rate. The kinematics are
\begin{equation}
\dot x=v\cos\theta,\qquad
\dot y=v\sin\theta,\qquad
\dot\theta=\omega ,
\label{eq:kinematics}
\end{equation}
where the forward speed $v>0$ is a constant maintained by a standard low-level velocity servo. We assume that $v$ belongs to a known compact interval, i.e., $v\in[v_{\min},v_{\max}]$, where $0<v_{\min}<v_{\max}$.
The yaw dynamics are torque-actuated and given by
\begin{equation}
J\dot\omega=-d_\omega\omega+\tau ,
\label{eq:yaw_dynamics}
\end{equation}
where $\tau\in\mathbb R$ is the yaw torque input, $J>0$ is the rotational inertia, and $d_\omega>0$ is the yaw damping coefficient. The exact values of $J$ and $d_\omega$ are unknown. We assume that the damping coefficient satisfies known bounds, i.e., $d_\omega\in [d_{\min},d_{\max}]$, with $0<d_{\min}<d_{\max}$.

The vehicle moves in an unknown, static, radially symmetric scalar field
\begin{equation}
\Phi(x,y)=\psi\left((x-x^\star)^2+(y-y^\star)^2\right),
\label{eq:source_field}
\end{equation}
where $(x^\star,y^\star)\in\mathbb R^2$ is the unknown source location. 
The radial profile is assumed to satisfy the following condition.

\begin{assumption}\label{ass:source_field}
    The function $\psi:\mathbb{R}_{\ge 0}\to\mathbb R$ is twice continuously differentiable, satisfies $\psi'(s)>0$ for all $s\ge0$, and satisfies $\psi(s)\to\infty$ as $s\to\infty$.
\end{assumption}

Under Assumption~\ref{ass:source_field}, the scalar field $\Phi$ has a unique minimum at the source, where the minimum value $\Phi(x^\star,y^\star)=\psi(0)$ is unknown. The task of the vehicle is to seek the minimum point of the scalar field. 

\begin{figure}
    \centering
    \includegraphics[width=0.45\linewidth]{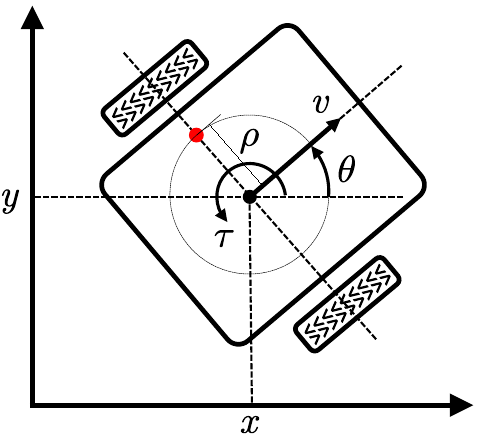}
    \caption{\unboldmath Geometry of the torque-controlled nonholonomic vehicle with a left-side scalar sensor at lateral offset $\rho$.}
    \label{fig:unicycle}
\end{figure}

A scalar sensor is mounted on the left side of the vehicle at a lateral offset $\rho>0$, as shown in Fig. \ref{fig:unicycle}. 
Define the position error vector in the body frame by
\begin{equation}\label{eq:error}
    \begin{bmatrix}
    x_e \\ y_e
    \end{bmatrix}:=
    \begin{bmatrix}
    \cos\theta & \sin\theta \\
    -\sin\theta & \cos\theta
    \end{bmatrix}
    \begin{bmatrix}
    x-x^\star \\ y-y^\star
    \end{bmatrix}.
\end{equation}
Hence, the value of the signal function $\psi$ at the position of the sensor is 
\begin{equation}
\ym=\psi(s),
\quad
s:=x_e^2+(y_e+\rho)^2.
\label{eq:measurement}
\end{equation}
The control objective is to design a yaw torque law using the real-time scalar measurement $\ym(t)$ such that the closed-loop vehicle approaches the source-centered circular motion. Measurements of the configuration variables $x,y$ or $\theta$ are not available.
The desired motion is characterized by $(x_e,y_e,\omega) =(0,-\rho,v/\rho)$.
At this motion, the vehicle travels on a circle of radius $\rho$ centered at the source, and the left-side sensor is located at the source, so $s=0$ and $\ym=\Phi(x^\star,y^\star)$.

\section{Control Designs and Main Results}
\label{sec:control}

This section presents two control designs with different sensing requirements. The first uses only the real-time scalar measurement $\ym(t)$ and combines a fast symmetric-product torque excitation with a slow Lie-bracket bias update, leading to a local practical stability result. The second additionally uses measurements of $v$ and $\omega$ to shape the yaw-rate tracking error directly, yielding a globally asymptotically stable averaged system and a semi-global practical stability result for the original system. Together, the two designs quantify the tradeoff between sensing capability and guaranteed stability region.

\subsection{Local Output-Feedback Design}

Let $u_1,u_2:\mathbb R_{\ge 0}\to\mathbb R$ be bounded, piecewise continuous, zero-mean, $T$-periodic functions. We assume that $u_1$ and $u_2$ are normalized as
\begin{equation}\label{eq:u_normalization}
    \frac{1}{T}\int_0^Tu_2(\sigma)\int_0^\sigma u_1(\xi){\rm d}\xi{\rm d}\sigma=\frac{1}{2}.
\end{equation}
Let $w:\mathbb R_{\ge 0}\to\mathbb R$ be a bounded, nonzero, piecewise continuous, zero-mean, $T_w$-periodic function. Let $W$ denote the $T_w$-periodic zero-mean antiderivative of $w$, i.e., $W'=w$. We assume that $w$ is scaled so that
\begin{equation}
\frac{1}{T_w}\int_0^{T_w}W^2(\sigma){\rm d}\sigma=\frac12 .
\label{eq:w_condition}
\end{equation}
Define the nominal steady-state yaw torque required to sustain the desired circular motion as
\begin{equation}
\mu^\star:=\frac{d_\omega v}{\rho}.
\label{eq:mustar}
\end{equation}
Let $\mathcal I:=[\mu_{\min},\mu_{\max}]$ be chosen, using the known bounds on $v$ and $d_\omega$, so that
\begin{equation}
0<\mu_{\min}<\frac{d_{\min}v_{\min}}{\rho}
\le \mu^\star
\le \frac{d_{\max}v_{\max}}{\rho}<\mu_{\max}.
\label{eq:bias_interval}
\end{equation}
Thus $\mu^\star\in\operatorname{Int}(\mathcal I)$, although $\mu^\star$ is unknown.

The scalar measurement $\ym$ is first passed through the washout filter
\begin{equation}
\dot z=-\lambda z+\lambda \ym,
\qquad
\yf=\ym-z,
\label{eq:filter}
\end{equation}
where $\lambda>0$. 
We propose the output-feedback control law
\begin{subequations}\label{eq:control_law}
\begin{align}
    \tau&=\mu + \frac{a}{\varepsilon}w\left(\frac{t}{\varepsilon}\right)H(\yf),\label{eq:controller_tau}\\
    \dot\mu&=
    \frac{b\Omega}{\sqrt{\delta}}
    \left[u_1\left(\frac{\Omega t}{\delta}\right)\sin(\ym)+u_2\left(\frac{\Omega t}{\delta}\right)\cos(\ym)\right],
    \label{eq:mu_update}
\end{align}
\end{subequations}
where $a>0$ and $b>0$ are control gains; $0<\varepsilon\ll1$, $0<\Omega\ll1$ 
determine the fast excitation and slow bias-tuning time scales, respectively;  
$\delta$ sets the averaging scale in the slow time $\Omega t$ and is selected as $\Omega\ll \delta \ll 1$; and $\mu$ is the slowly varying yaw-torque bias.
The shaping function $H:\mathbb R\to\mathbb R$ satisfies $H\in C^2(\mathbb R)$ and
\begin{equation}
\Gamma(q):=H(q)H'(q)>0,\quad \Gamma'(q)>0, \quad \forall q\in\mathbb R .
\label{eq:H_condition}
\end{equation}
The control law \eqref{eq:filter}--\eqref{eq:control_law} is output feedback: it uses the scalar measurement $\ym$ only, together with the internally generated filter and bias state.

For analysis, define
\begin{equation}
r:=v-\rho\omega.
\label{eq:r_def}
\end{equation}
Then, the closed-loop system can be written in the analysis coordinates $(x_e,y_e,r,z,\mu)$ as
\begin{subequations}
\label{eq:closed_loop_original}
\begin{align}
\dot x_e&=
r+\omega(y_e+\rho),
\label{eq:closed_loop_original_xe}\\
\dot y_e&=
-\omega x_e,
\label{eq:closed_loop_original_ye}\\
\dot r&=
-\frac{d_\omega}{J}r
-\frac{\rho}{J}
\left[
\mu-\mu^\star
+\frac{a}{\varepsilon}w\left(\frac{t}{\varepsilon}\right)H(\yf)
\right],
\label{eq:closed_loop_original_r}\\
\dot z&=
-\lambda z+\lambda \psi(s),
\label{eq:closed_loop_original_z}\\
\dot\mu&=\frac{b\Omega}{\sqrt{\delta}}
    \left[u_1\left(\frac{\Omega t}{\delta}\right)\sin(\ym)+u_2\left(\frac{\Omega t}{\delta}\right)\cos(\ym)\right],
\label{eq:closed_loop_original_mu}
\end{align}
\end{subequations}
where 
$\omega=(v-r)/\rho$, $s=x_e^2+(y_e+\rho)^2$, and $\yf=\psi(s)-z$. 

The closed-loop system \eqref{eq:closed_loop_original} contains three separated time scales. 
The oscillatory torque term $w(t/\varepsilon)$ evolves on the fast time scale $O(\varepsilon)$. 
The plant and filter variables $x_e,y_e,r,z$ evolve on the normal time scale $O(1)$. 
The bias state $\mu$ evolves on a much slower time scale induced by the small parameter $\Omega$. 
The above time-scale separation motivates the two-step averaging analysis used below. 
In the fast averaging calculations, the bias state $\mu$ is treated as \textit{frozen} because its variation over the $O(\varepsilon)$ excitation period is negligible. 
After this fast averaging step, the plant/filter variables $x_e,y_e,r,z$ evolve on the normal time scale, while $\mu$ varies on the slower bias-adaptation time scale. 
Therefore, in the slow bias averaging, the fixed-bias averaged dynamics induce a \textit{steady-state input-output map} from $\mu$ to the corresponding steady scalar measurement. 
The slow averaging is then applied to this steady map.

\subsubsection{Fast Averaging via Symmetric-Product Approximation}

In the first step of the analysis, we average the closed-loop system \eqref{eq:closed_loop_original} on the fast time scale using the \textit{symmetric-product approximation} for the yaw-torque channel \cite{bullo2002averaging}. During this fast averaging calculation, the controller states $z$ and $\mu$ are treated as frozen parameters.

The relevant mechanical channel for the fast averaging is the one-dimensional yaw dynamics
\begin{equation}
\nabla_{\dot\theta}\dot\theta
=Y_0(\dot{\theta})
+
\frac{a}{\varepsilon}
w\left(\frac{t}{\varepsilon}\right)
Y_1(\theta),
\label{eq:yaw_mechanical_channel}
\end{equation}
where 
\begin{equation*}
Y_0(\dot{\theta}):=\left(
-\frac{d_\omega}{J}\dot\theta+\frac{\mu}{J}
\right)\frac{\partial}{\partial\theta},\quad
Y_1(\theta)
:=
\frac{1}{J}H(\psi(s)-z)\frac{\partial}{\partial\theta},
\end{equation*}
and $\nabla$ is the trivial affine connection on $\mathbb{R}$. Strictly speaking, $Y_1$ also depends on the frozen parameters $x$, $y$,
and $z$ through $s=x_e^2+(y_e+\rho)^2$; this dependence is suppressed for
notational simplicity, while the dependence of $s$ on
$\theta$ is retained in the symmetric-product calculation.

Consider the change of variables
\begin{equation}\label{eq:change-of-variable-yaw}
    \tilde{\omega}=\omega-\frac{a}{J}W\left(\frac{t}{\varepsilon}\right)H(\psi(s)-z).
\end{equation}
For fixed values of the frozen parameters, following \cite{bullo2002averaging}, the solutions of \eqref{eq:yaw_mechanical_channel} in the variables \eqref{eq:change-of-variable-yaw} approximate the solutions of
\begin{equation}
\nabla_{\dot{\bar{\theta}}}\dot{\bar{\theta}}
=Y_0(\dot{\bar{\theta}})
-\frac{a^2}{4}
\langle Y_1:Y_1\rangle(\bar\theta),
\label{eq:yaw_mechanical_averaged}
\end{equation}
where $\langle Y_1:Y_1\rangle$ is the symmetric product given by $\langle Y_1:Y_1\rangle:=2\nabla_{Y_1} {Y_1}$.
Since the connection is trivial, direct calculation yields
\begin{equation}
    \nabla_{Y_1} {Y_1}(\theta) =
\frac{1}{J^2}
H(\yf)H'(\yf)
\frac{\partial\yf}{\partial\theta}
\frac{\partial}{\partial\theta}.
\end{equation}
Hence, using $\partial\yf/\partial\theta=-2\rho\psi'(s)x_e$, \eqref{eq:yaw_mechanical_averaged} becomes
\begin{equation}
    J\dot{\bar{\omega}}
=-d_\omega \bar{\omega}+\mu+
\frac{a^2\rho}{J}
\Gamma(\baryf)\psi'(\bar s)\bar{x}_e,
\end{equation}
where $\bar s:=\bar{x}_e^2+(\bar{y}_e+\rho)^2$ and $\baryf:=\psi(\bar s)-\bar{z}$.

With the definition $\bar r:=v-\rho\bar\omega$, one has
$\dot{\bar r}=-\rho\dot{\bar\omega}$. Hence, for each fixed bias
$\mu$, the associated fixed-bias fast-averaged system is
\begin{subequations}
\label{eq:fast_averaged}
\begin{align}
\dot{\bar x}_e&=
\bar r+\bar\omega(\bar y_e+\rho),
\label{eq:fast_averaged_xe}\\
\dot{\bar y}_e&=
-\bar\omega\bar x_e,
\label{eq:fast_averaged_ye}\\
\dot{\bar r}&=
-\frac{d_\omega}{J}\bar r
-\frac{\rho}{J}(\mu-\mu^\star)
-\frac{a^2\rho^2}{J^2}
\Gamma(\baryf)\psi'(\bar s)\bar x_e,
\label{eq:fast_averaged_r}\\
\dot{\bar z}&=-\lambda\bar z+\lambda\psi(\bar s),\label{eq:fast_averaged_z}
\end{align}
\end{subequations}
where $\bar{\omega}=(v-\bar{r})/\rho$. For each fixed $\mu\in\operatorname{Int}(\mathcal{I})$, the fast-averaged system \eqref{eq:fast_averaged} has the unique equilibrium
\begin{equation}\label{eq:equilibrium}
    \bar x_{e*}=0,\;
    \bar y_{e*} = -\frac{\rho\mu^\star}{\mu},\;
    \bar r_*=-\frac{\rho}{d_\omega}(\mu-\mu^\star),\;
    \bar z_*=\psi(\bar s_*),
\end{equation}
where 
\begin{equation*}
\bar s_*(\mu)=
\bar x_{e*}^2+\left(\bar y_{e*}+\rho\right)^2 
=\rho^2\left(1-\frac{\mu^\star}{\mu}\right)^2.
\end{equation*}

\begin{proposition}
\label{prop:fixed_bias_les}
Suppose that Assumption~1 holds. Then, for each fixed
$\mu\in\operatorname{Int}(\mathcal I)$, the equilibrium
\eqref{eq:equilibrium} of \eqref{eq:fast_averaged} is locally
exponentially stable.
\end{proposition}

\begin{proof}
Fix $\mu\in\operatorname{Int}(\mathcal I)$ and introduce the local
linearization coordinates $\chi:=
\operatorname{col}\left(
\bar x_e-\bar x_{e*},
\bar y_e-\bar y_{e*},
\bar r-\bar r_*,
\bar z-\bar z_*
\right)$.
At the equilibrium \eqref{eq:equilibrium}, one has $\bar\omega_*={\mu}/{d_\omega}$ and $\bar{\mathsf{y}}_{\mathsf{f}*}=\psi(\bar s_*)-\bar z_*=0$.
Linearizing \eqref{eq:fast_averaged} about the equilibrium \eqref{eq:equilibrium} gives $\dot\chi=A_\mu\chi$, where 
\begin{equation*}
A_\mu=
\begin{bmatrix}
0 & \dfrac{\mu}{d_\omega} & \dfrac{\mu^\star}{\mu} & 0 \\[1mm]
-\dfrac{\mu}{d_\omega} & 0 & 0 & 0 \\[1mm]
-\dfrac{a^2\rho^2}{J^2}\Gamma(0)\psi'(\bar s_*) & 0
& -\dfrac{d_\omega}{J} & 0 \\[1mm]
0 & 2\lambda\psi'(\bar s_*)(\bar y_{e*}+\rho) & 0 & -\lambda
\end{bmatrix}.
\end{equation*}
The characteristic polynomial of $A_\mu$ is
\begin{equation}
\det(\zeta I-A_\mu)
=
(\zeta+\lambda)
\left[
\zeta^3+c_2\zeta^2+c_1\zeta+c_0
\right],
\end{equation}
where $c_2:={d_\omega}/{J}>0$, $c_0:= \mu^2/(Jd_\omega)>0$, and
\begin{equation*}
    c_1:=\left(\frac{\mu}{d_\omega}\right)^2+\frac{\mu^\star}{\mu}\frac{a^2\rho^2}{J^2}\Gamma(0)\psi'(\bar s_*)>0.
\end{equation*}
Moreover, $c_2c_1-c_0=d_\omega \mu^\star a^2\rho^2\Gamma(0)\psi'(\bar s_*)/(\mu J^3)>0$.
Hence, by the Routh--Hurwitz criterion for cubic polynomials, all roots
of $\zeta^3+c_2\zeta^2+c_1\zeta+c_0$ have negative real parts. The
remaining eigenvalue is $-\lambda<0$. Therefore $A_\mu$ is Hurwitz.
The linearization theorem then implies that the equilibrium
\eqref{eq:equilibrium} of \eqref{eq:fast_averaged} is locally
exponentially stable.\qed
\end{proof}

\subsubsection{Slow Bias Averaging via Lie-Bracket Approximation}

Next, we derive the slow-averaged dynamics of $\mu$. 
Locally around the exponentially stable equilibrium branch, \eqref{eq:fast_averaged} defines a steady input-output map $\mu\mapsto P(\mu)$ from the yaw-torque bias $\mu$ to the corresponding steady scalar measurement. 
The slow update \eqref{eq:mu_update} is designed as a \textit{Lie-bracket extremum-seeking} law on the slow time $\sigma:=\Omega t$ \cite{DurrStankovicEbenbauerJohansson2013,durr2017extremum}. As shown below, the resulting averaged bias dynamics are a gradient descent for $P$.

Define the steady input-output map by
\begin{equation}\label{eq:P_def}
P:\mathcal{I}\to\mathbb{R},\qquad P(\mu):=\psi\left(\bar s_*(\mu)\right).
\end{equation}
Differentiating $P$ gives
\begin{equation}
P'(\mu)=p(\mu)(\mu-\mu^\star),\quad p(\mu):=\frac{2\rho^2\mu^\star}{\mu^3}\psi'(\bar s_*(\mu)).
\label{eq:P_derivative}
\end{equation}
Since $\mathcal I\subset(0,\infty)$ is compact, $\mu^\star>0$, and $\psi'(s)>0$, one has 
\begin{equation}
    \underline p:=\min_{\mu\in \mathcal{I}} p(\mu)>0,
\end{equation}
and
\begin{equation}
(\mu-\mu^\star)P'(\mu)
\ge
\underline p(\mu-\mu^\star)^2,
\quad
\forall \mu\in\mathcal I.
\label{eq:P_monotonicity}
\end{equation}
Thus, $\mu^\star$ is the unique minimizer of $P$ on $\mathcal I$, and the descent direction $-P'(\mu)$ points toward $\mu^\star$.

Now we derive the slow-averaged bias dynamics. Let
\begin{equation}
\sigma:=\Omega t.
\end{equation}
Restricting the slow bias update \eqref{eq:mu_update} to the equilibrium family gives the reduced system
\begin{equation}
    \frac{{\rm d}\mu}{{\rm d}\sigma}=\frac{b}{\sqrt{\delta}}\left[u_1\left(\frac{\sigma}{\delta} \right)\sin(P(\mu))+u_2 \left(\frac{\sigma}{\delta} \right)\cos(P(\mu)) \right],
\end{equation}
where, along this equilibrium family, the steady scalar measurement is ${\bar{\mathsf{y}}_{\mathsf{m}*}}=P(\mu)$.
Using the normalization \eqref{eq:u_normalization}, following \cite{DurrStankovicEbenbauerJohansson2013}, the corresponding Lie-bracket averaged dynamics are
\begin{equation}
    \frac{{\rm d}\bar{\mu}}{{\rm d}\sigma}
    =\frac{b^2}{2}[g_1,g_2](\bar{\mu})
    =-\frac{b^2}{2}P'(\bar\mu),
\end{equation}
where $g_1(\mu):=\sin(P(\mu))$ and $g_2(\mu):=\cos(P(\mu))$.
Equivalently, in the original time variable,
\begin{equation}\label{eq:bar_mu}
    \dot{\bar{\mu}}=-\frac{b^2\Omega}{2}P'(\bar\mu).
\end{equation}
The averaged bias dynamics \eqref{eq:bar_mu} are a gradient descent for $P$. From \eqref{eq:P_monotonicity} and \eqref{eq:bar_mu}, the Lyapunov function
\begin{equation}\label{eq:V_mu}
    V_\mu:=\frac{1}{2}(\bar\mu-\mu^\star)^2
\end{equation}
satisfies
\begin{equation}\label{eq:dot_V_mu}
    \dot V_\mu=-\frac{b^2\Omega}{2}(\bar\mu-\mu^\star)P'(\bar\mu)
    \le-\frac{b^2\Omega}{2}\underline p(\bar\mu-\mu^\star)^2,
\end{equation}
which implies that the reduced slow-averaged bias dynamics drive $\bar\mu$ to $\mu^\star$.


\subsubsection{Stability Analysis}

For notation simplicity, denote $\xi:=(\bar x_e,\bar y_e,\bar r,\bar z)$.
Let $f_\xi(\xi,\mu)$ denote the right-hand side of \eqref{eq:fast_averaged}. 
Then, combining the fast averaged dynamics \eqref{eq:fast_averaged} with the slow averaged bias dynamics \eqref{eq:bar_mu} yields
\begin{subequations}
\label{eq:averaged_interconnection}
\begin{align}
\dot{\xi}&=f_\xi(\xi,\bar{\mu}),\\
\dot{\bar\mu}&=
- \Omega k_{\mu} P'(\bar\mu),
\label{eq:averaged_interconnection_mu}
\end{align}
\end{subequations}
where $k_{\mu}:={b^2}/{2}$. The equilibrium of \eqref{eq:averaged_interconnection} is given by
\begin{equation}
(\xi^\star,\mu^\star):=(0,-\rho,0,\psi(0),\mu^\star).
\label{eq:averaged_interconnection_equilibrium}
\end{equation}

\begin{proposition}
\label{prop:local_composite_average}
Suppose that Assumption~\ref{ass:source_field} holds.
Then there exists $\Omega_0>0$ such that, for every
$\Omega\in(0,\Omega_0]$, the equilibrium $(\xi^\star,\mu^\star)$ of
\eqref{eq:averaged_interconnection} is locally asymptotically stable.
\end{proposition}

\begin{proof}
In the slow time $\sigma:=\Omega t$, \eqref{eq:averaged_interconnection}
can be written in the standard singularly perturbed form \cite{Khalil:1173048}.
By Proposition~\ref{prop:fixed_bias_les}, for each fixed
$\mu\in\operatorname{Int}(\mathcal I)$, the equilibrium
$\xi_*(\mu)=(\bar{x}_{e*},\bar{y}_{e*},\bar{r}_{*},\bar{z}_{*})$ of the boundary-layer system is locally exponentially stable. Since the linearization depends continuously on $\mu$, this local exponential stability is uniform for $\mu$ in a sufficiently small compact neighborhood of $\mu^\star$ contained in $\operatorname{Int}(\mathcal{I})$.
The reduced system \eqref{eq:averaged_interconnection_mu} is obtained by setting $\xi=\xi_*(\mu)$. Following \eqref{eq:V_mu}--\eqref{eq:dot_V_mu}, $\mu^\star$ is a locally exponentially stable equilibrium of \eqref{eq:averaged_interconnection_mu}.
The conclusion follows from the standard Tikhonov singular perturbation
theorem on the infinite time interval \cite[Theorem 11.4]{Khalil:1173048}. \qed
\end{proof}

Following \eqref{eq:change-of-variable-yaw}, define the change of variables
\begin{subequations}\label{eq:change-of-variable}
   \begin{align}
    \tilde{x}_e&=x_e,\quad \tilde{y}_e=y_e,\quad \tilde{z}=z,\quad \tilde{\mu}=\mu\\
    \tilde{r}&=r+\frac{\rho a}{J}W\left(\frac{t}{\varepsilon}\right)H(\psi(s)-z).
   \end{align}
\end{subequations}
By the symmetric-product approximation \cite{bullo2002averaging} and the Lie-bracket averaging theorem \cite{DurrStankovicEbenbauerJohansson2013,durr2017extremum}, for sufficiently small $\varepsilon$, $\delta$, and $\Omega/\delta$, the solutions of \eqref{eq:closed_loop_original} in the variables of \eqref{eq:change-of-variable} approximate the solutions of \eqref{eq:averaged_interconnection} on compact time intervals.
In the terminology of Definition \ref{def:1}, we state the main result as a consequence of Propositions \ref{prop:local_composite_average} and Lemma \ref{lem:1}(i).

\begin{theorem}
\label{thm:main_practical_interconnection}
Suppose Assumption~\ref{ass:source_field} holds. Choose the shaping function $H$ such that \eqref{eq:H_condition} is satisfied. 
Then, there exist constants
$\Omega_0>0$, $\delta_0>0$, $\gamma_0>0$, and $\varepsilon_0>0$ such that, for every
$\Omega\in(0,\Omega_0]$, every $\delta\in(0,\delta_0]$ satisfying
$\Omega/\delta\le \gamma_0$, and every $\varepsilon\in(0,\varepsilon_0]$, the equilibrium $(\xi^\star,\mu^\star)$ is locally practically uniformly asymptotically stable for \eqref{eq:closed_loop_original} in the variables of \eqref{eq:change-of-variable}.
\end{theorem}

\begin{remark}
The output-feedback design is not a minor modification of the torque law in \cite{suttner2022source}.
In \cite{suttner2022source}, the torque consists only of a large-amplitude oscillatory term, and the limiting circular orbit is determined implicitly by the balance among the physical parameters, the design gains, and the averaged gradient-dependent torque. In the present design, the yaw torque contains an additional slowly varying bias. This bias is not only a constant offset; it is the quantity to be tuned so that the desired source-centered orbit is selected. The lateral sensor placement is also essential. Instead of placing the sensor in front of the vehicle, the sensor is placed on the side, so that the desired circular motion corresponds to the sensor being located at the source. Thus, the scalar measurement provides information for tuning the bias toward the source-centered orbit. With only the scalar measurement, this tuning is achieved indirectly through a Lie-bracket extremum-seeking update, which leads to a local practical stability result.
\end{remark}

\subsection{Semi-Global Velocity-Assisted Design}

The preceding design uses only the scalar measurement $\ym$. In many robotic platforms, however, the forward speed $v$ and yaw rate $\omega$ are readily available from wheel, motor, or inertial sensors. The following design uses this additional velocity information to tune the bias torque; it still does not require measurements of the vehicle position, heading, source location, field gradient, inertia, or damping coefficient.

We keep the washout filter \eqref{eq:filter} and the same fast torque excitation, but replace the Lie-bracket bias update with the yaw-rate tracking feedback
\begin{subequations}\label{eq:control2}
\begin{align}
    \tau&=\mu + \frac{a}{\varepsilon}w\left(\frac{t}{\varepsilon}\right)H(\yf),\label{eq:controller_tau2}\\
    \dot\mu&=k(v-\rho \omega),
    \label{eq:mu_update2}
\end{align}
\end{subequations}
where $a>0$ and $k>0$ are control gains. Then, the closed-loop system can be written in the analysis coordinates $(x_e,y_e,r,z,\mu)$ as
\begin{subequations}
\label{eq:closed_loop_original2}
\begin{align}
\dot x_e&=
r+\omega(y_e+\rho),
\label{eq:closed_loop_original_xe2}\\
\dot y_e&=
-\omega x_e,
\label{eq:closed_loop_original_ye2}\\
\dot r&=
-\frac{d_\omega}{J}r
-\frac{\rho}{J}
\left[
\mu-\mu^\star
+\frac{a}{\varepsilon}w\left(\frac{t}{\varepsilon}\right)H(\yf)
\right],
\label{eq:closed_loop_original_r2}\\
\dot z&=
-\lambda z+\lambda \psi(s),
\label{eq:closed_loop_original_z2}\\
\dot\mu&=kr,
\label{eq:closed_loop_original_mu2}
\end{align}
\end{subequations}
After the same symmetric-product averaging step as in the previous subsection, the averaged closed-loop system is
\begin{subequations}
\label{eq:fast_averaged2}
\begin{align}
\dot{\xi}&=f_\xi(\xi,\bar{\mu})\\
\dot{\bar \mu}&=k\bar{r}.
\end{align}
\end{subequations}
The equilibrium of \eqref{eq:fast_averaged2} is given by \eqref{eq:averaged_interconnection_equilibrium}.

\begin{proposition}
\label{prop:averaged_interconnection_stability}
Suppose Assumption~\ref{ass:source_field} holds. Then the equilibrium $(\xi^\star,\mu^\star)$ 
is globally asymptotically stable for \eqref{eq:fast_averaged2}. 
\end{proposition}

\begin{proof}
Let $\kappa:=a^2\rho^2/J^2$. Consider
\begin{equation}
\begin{aligned}
V&:=\frac12\bar r^2+\frac{\kappa\Gamma(0)}{2}\big(\psi(\bar s)-\psi(0)\big)\\
&\quad+\frac{\kappa}{2}\int_0^{\baryf}\left(\Gamma(\zeta)-\Gamma(0)\right){\rm d}\zeta+\frac{\rho}{2Jk}(\bar\mu-\mu^\star)^2.
\end{aligned}
\end{equation}
Since $\psi'(s)>0$ and $\Gamma'(\zeta)>0$, $V$ is positive definite
with respect to $(\xi^\star,\mu^\star)$ and radially unbounded. Differentiating $V$ along \eqref{eq:fast_averaged2}
and using $\dot{\bar s}=2\bar x_e\bar r$ and $\dot{\bar{\mathsf{y}}}_{\mathsf{f}}=2\psi'(\bar s)\bar x_e\bar r-\lambda\baryf$ yield
\begin{equation}
\dot V
=
-\frac{d_\omega}{J}\bar r^2
-\frac{\kappa\lambda}{2}
\left(\Gamma(\baryf)-\Gamma(0)\right)\baryf\le 0.
\label{eq:Vxi_dot}
\end{equation}
Since $\Gamma$ is strictly increasing, $(\Gamma(\baryf)-\Gamma(0))\baryf=0$ if and only if $\baryf=0$. Hence $\dot V=0$ if and only if $\bar r=0$ and $\baryf=0$. Let $\mathcal M$ be the largest invariant set contained in $\{(\xi,\bar\mu):\bar r=0, \baryf=0\}$. On $\mathcal M$, one has $\dot{\bar\mu}=k\bar r=0$, $\bar\omega=v/\rho$, and $\dot{\bar s}=2\bar x_e\bar r=0$. Invariance of $\bar r=0$ gives $0=-(\rho/J)(\bar\mu-\mu^\star)-\kappa\Gamma(0)\psi'(\bar s)\bar x_e$. Since $\bar\mu$ and $\bar s$ are constant on $\mathcal M$, differentiating this identity along $\mathcal M$ yields $0=-\kappa\Gamma(0)\psi'(\bar s)(v/\rho)(\bar y_e+\rho)$. Thus, $\bar y_e=-\rho$. Invariance of $\bar y_e=-\rho$ gives $\bar x_e=0$, and the preceding identity gives $\bar\mu=\mu^\star$. Finally, $\baryf=0$ and $\bar s=0$ imply $\bar z=\psi(0)$. Therefore $\mathcal M$ is the singleton $(\xi^\star,\mu^\star)$. By the invariance principle, the equilibrium is globally attractive. Together with global stability, this proves global asymptotic stability.\qed
\end{proof}

By the symmetric-product approximation \cite{bullo2002averaging}, the solutions of \eqref{eq:closed_loop_original2} in the variables of \eqref{eq:change-of-variable} approximate the solutions of \eqref{eq:fast_averaged2} on compact time intervals.
In the terminology of Definition \ref{def:2}, we state the main result as a consequence of Propositions \ref{prop:averaged_interconnection_stability} and Lemma \ref{lem:1}(ii).

\begin{theorem}
\label{thm:main_practical_interconnection2}
Suppose Assumption~\ref{ass:source_field} holds. Choose the shaping function $H$ such that \eqref{eq:H_condition} is satisfied. Then, the equilibrium $(\xi^\star,\mu^\star)$ is semi-globally practically uniformly asymptotically stable for \eqref{eq:closed_loop_original2} in the variables of \eqref{eq:change-of-variable}.
\end{theorem}

\begin{remark}
The velocity-assisted design shows how the stability region can be enlarged when velocity measurements are available. The additional measurement of the yaw-rate tracking error allows the bias to be adjusted directly, instead of being tuned through the steady input-output map used in the output-feedback design. This removes the need to rely only on local exponential stability of the fixed-bias averaged dynamics. As a result, the averaged closed-loop system admits a global Lyapunov function. The conclusion for the original system remains semi-global and practical because it is obtained through averaging. Nevertheless, the result shows that the source-centered orbit can be practically stabilized from arbitrarily large prescribed compact sets without measuring the vehicle position, heading, source location, field gradient, inertia, or damping coefficient.
\end{remark}

\section{Numerical Simulations}
\label{sec:simulation}

This section illustrates the proposed torque-tuning designs through numerical simulations. All parameters are given in SI units. The source is placed at the origin, i.e., $(x^\star,y^\star)=(0,0)$, and the scalar field is chosen as
\begin{equation}
\Phi(x,y)=\psi(x^2+y^2),
\label{eq:sim_field}
\end{equation}
where $\psi(s):=s$. The physical parameters are $J=0.06$, $d_\omega=0.12$, $\rho=0.15$, and $v=0.8$. Hence, $\mu^\star=d_\omega v/\rho=0.64$. The estimated bounds are $d_{\min}=0.01$, $d_{\max}=0.2$, $v_{\min}=0.5$, and $v_{\max}=1$, which give $\mathcal I=[0.03,1.4]$.

For the output-feedback design \eqref{eq:control_law}, the controller parameters are selected as $a=0.2$, $\varepsilon=0.02$, $b=1$, $\Omega=0.005$, $\delta=0.2$, and $\lambda=2$. For the velocity-assisted design \eqref{eq:control2}, the same parameters are used, with the additional gain $k=0.0015$. In both cases, $\mu(0)=0.05$, and the washout filter is initialized as $z(0)=\ym(0)$. The excitation signals are chosen as $u_1=\cos$, $u_2=\sin$, and $w=\sin$, which satisfy \eqref{eq:u_normalization} and \eqref{eq:w_condition}. The shaping function is $H(q):=3e^{q/30}$, which satisfies \eqref{eq:H_condition}. The initial vehicle state is $(x(0),y(0),\theta(0),\omega(0))=(10,10,0,0)$.

Figures~\ref{fig:sim_traj} and \ref{fig:sim_signals} show the response under the output-feedback design \eqref{eq:control_law}. The vehicle center approaches a small orbit around the source, and the scalar measurement $\ym(t)$ decreases to a small neighborhood of zero. Since the theoretical guarantee for this design is local, this simulation illustrates performance beyond the certified region.

Figures~\ref{fig:sim_traj2} and \ref{fig:sim_signals2} show the response under the velocity-assisted design \eqref{eq:control2}. Compared with the output-feedback design, the transient is shorter, and the vehicle reaches the source-centered orbit more directly. The scalar measurement also converges faster to a neighborhood of zero.

Figure~\ref{fig:mu} uses a longer horizon to show the slow bias $\mu$ evolution. The dashed line denotes $\mu^\star$. The velocity-assisted design drives $\mu(t)$ toward $\mu^\star$, as predicted by the averaged analysis. The output-feedback design exhibits larger slow oscillations in $\mu(t)$, which is consistent with the Lie-bracket bias-tuning mechanism.

\begin{figure}
\centering
\includegraphics[width=0.8\columnwidth]{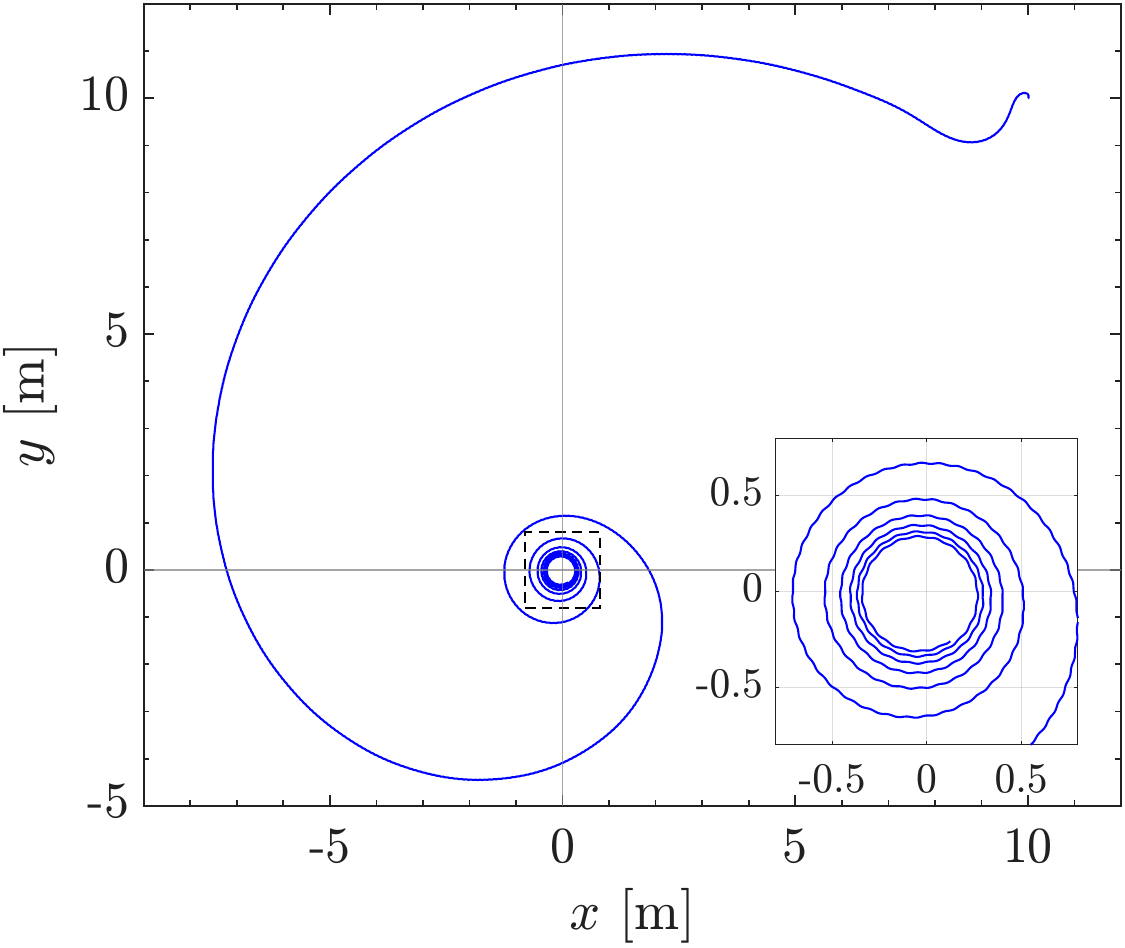}
\caption{Planar trajectory of the vehicle center under the output-feedback design \eqref{eq:control_law}. The source is located at the origin.}
\label{fig:sim_traj}
\end{figure}

\begin{figure}
\centering
\includegraphics[width=0.85\columnwidth]{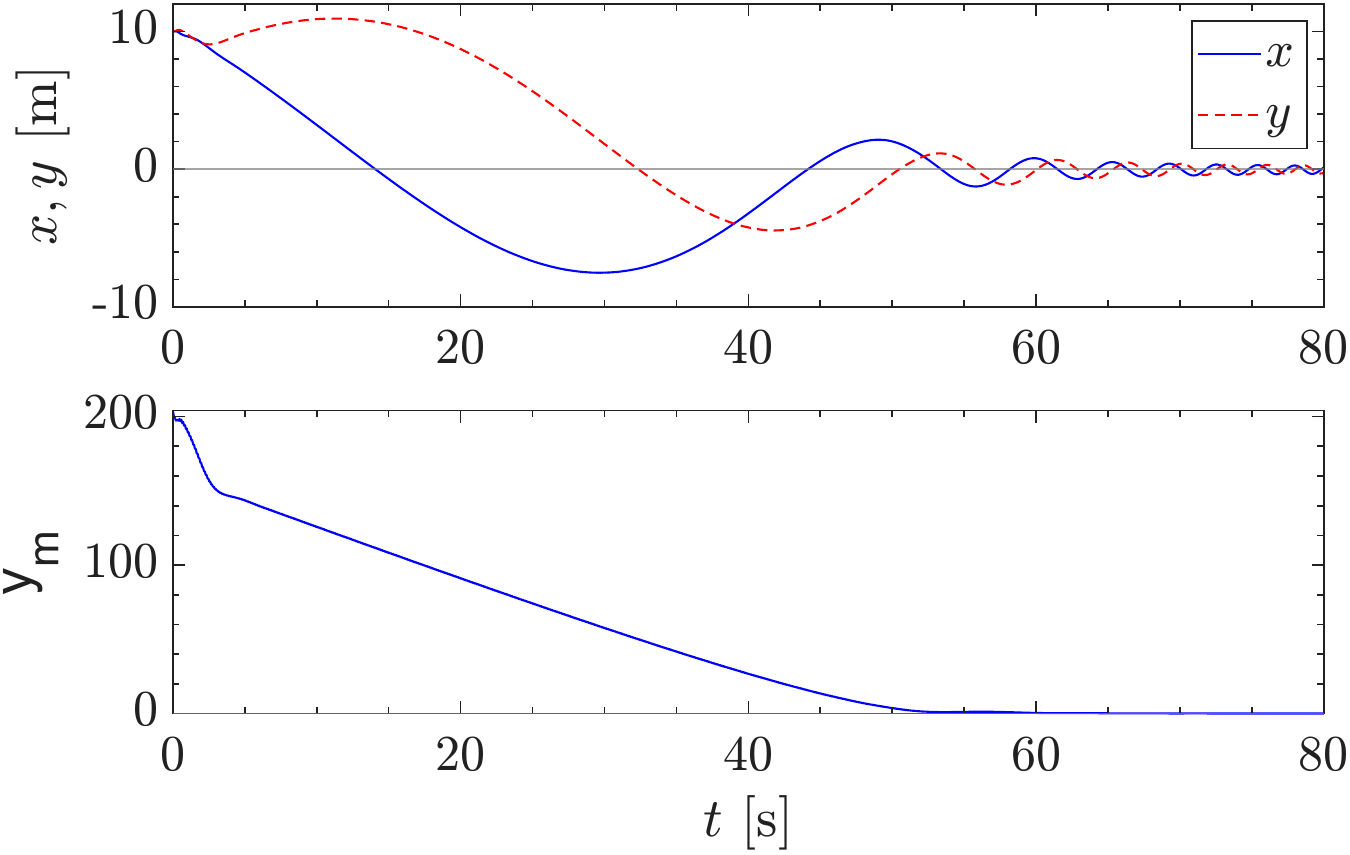}
\caption{\unboldmath Closed-loop position coordinates and scalar measurement under the output-feedback design \eqref{eq:control_law}.}
\label{fig:sim_signals}
\end{figure}

\begin{figure}
\centering
\includegraphics[width=0.8\columnwidth]{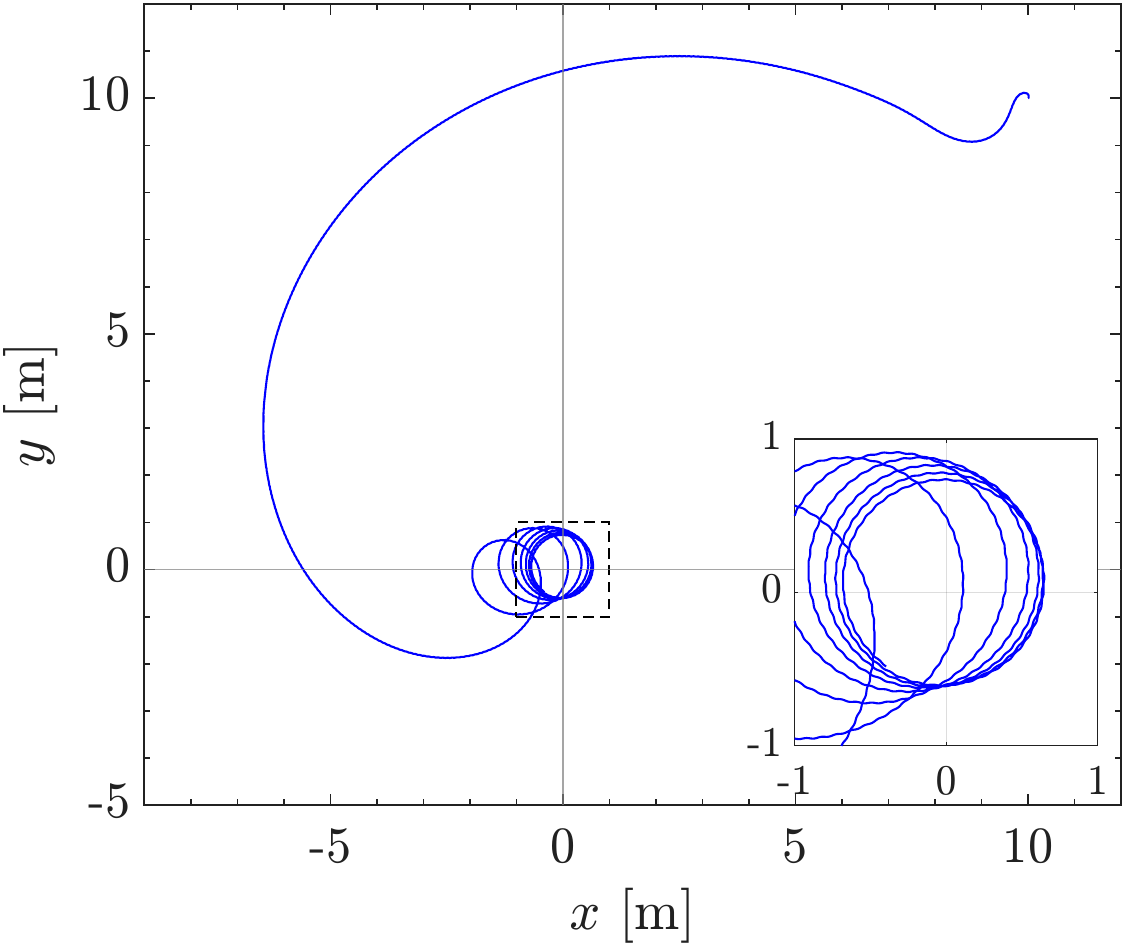}
\caption{Planar trajectory of the vehicle center under the velocity-assisted design \eqref{eq:control2}. The source is located at the origin.}
\label{fig:sim_traj2}
\end{figure}

\begin{figure}
\centering
\includegraphics[width=0.85\columnwidth]{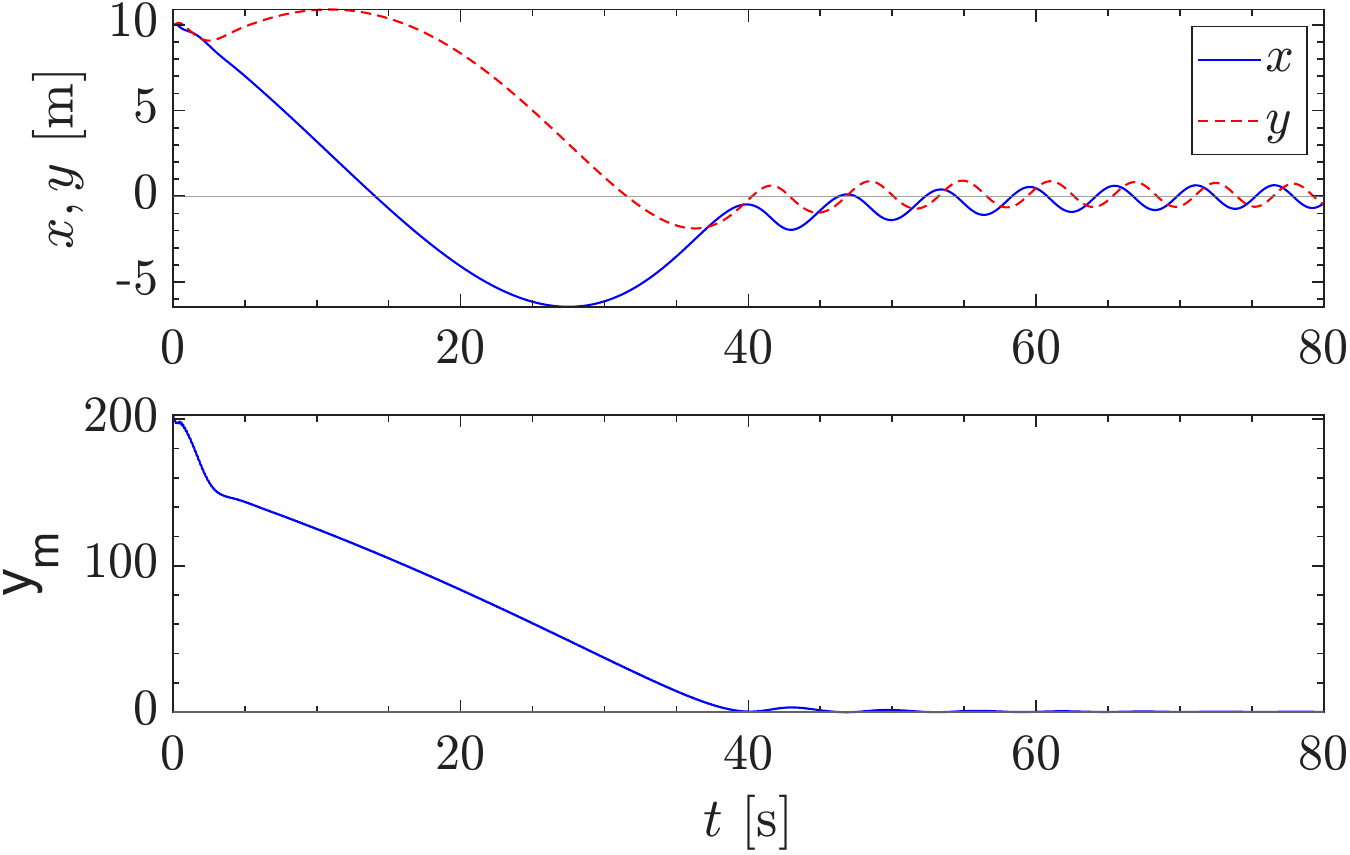}
\caption{\unboldmath Closed-loop position coordinates and scalar measurement under the velocity-assisted design \eqref{eq:control2}.}
\label{fig:sim_signals2}
\end{figure}

\begin{figure}
\centering
\includegraphics[width=0.85\columnwidth]{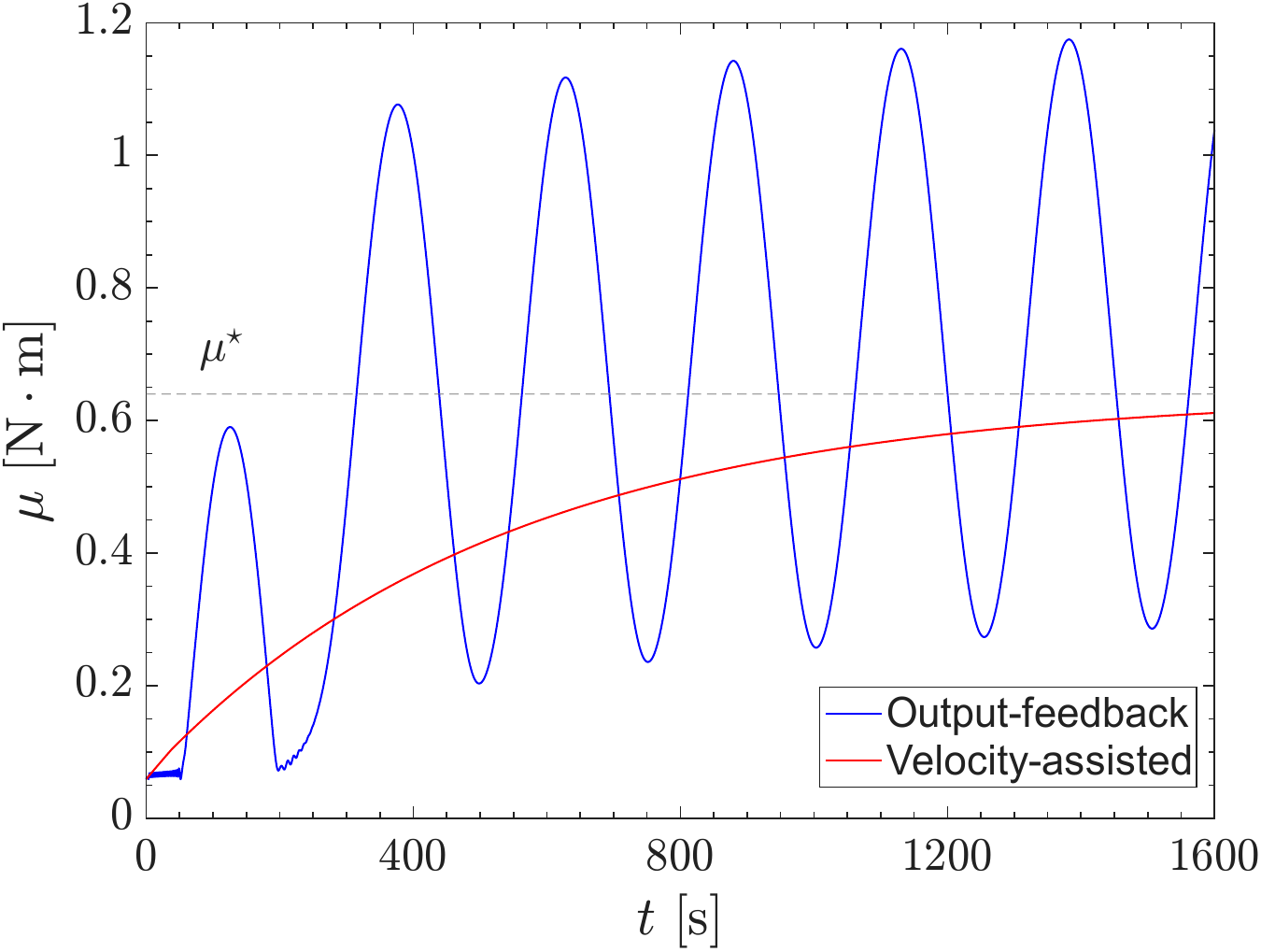}
\caption{\unboldmath Time histories of the yaw-torque bias $\mu(t)$ under the output-feedback and velocity-assisted designs. The dashed line denotes the nominal bias $\mu^\star=0.64$.}
\label{fig:mu}
\end{figure}

\section{Conclusion}
\label{sec:conclusion}

This paper developed two torque-tuning source-seeking designs for a nonholonomic vehicle with a laterally displaced scalar sensor. A fast oscillatory torque input was used to generate an averaged steering effect through a symmetric-product approximation. For the output-feedback design, a slow Lie-bracket bias update was introduced using only the real-time measurement $\ym(t)$, leading to local practical stability. When measurements of $v$ and $\omega$ are additionally available, the bias was tuned through the yaw-rate tracking error. This velocity-assisted design yields a globally asymptotically stable averaged system and, consequently, semi-global practical stability for the original closed-loop system. 
Future work will address robustness to measurement noise, actuator saturation, and experimental validation on robotic platforms.




\bibliographystyle{model1-num-names}

\bibliography{v1-main}

\end{document}